\documentclass[11pt,a4paper]{amsart}

\pdfoutput=1

\usepackage{natbib}

\usepackage{tikz}
\usepackage{mathrsfs}
\usepackage{amsfonts}

\usepackage{upgreek}
\usepackage{float}

\DeclareFontFamily{OT1}{pzc}{}
\DeclareFontShape{OT1}{pzc}{m}{it}{<-> s * [1.10] pzcmi7t}{}
\DeclareMathAlphabet{\mathpzc}{OT1}{pzc}{m}{it}

\usepackage{geometry}                
\usepackage{graphicx}
\usepackage{amssymb}
\usepackage{epstopdf}
\usepackage{fullpage}
\usepackage{color}
\usepackage{float}
\usepackage{caption}
\usepackage{tablefootnote}

\usepackage{url}

\RequirePackage[OT1]{fontenc}
\RequirePackage{amsthm,amsmath}
\usepackage{amsmath, amssymb, amsfonts} 
\usepackage{algorithm}
\usepackage{algpseudocode}
\usepackage{booktabs}  
\usepackage{graphicx, subfigure} 
\usepackage{multirow} 

\usepackage{mathabx}

\usepackage{longtable}

\usepackage{accents}

\usepackage{tikz}
\usetikzlibrary{bayesnet}
\usetikzlibrary{arrows}
\usepackage{graphicx}
\usepackage{caption}
\usepackage{subcaption}
\usepackage{amsmath,amsthm,amssymb,amsfonts,bbm}
\usetikzlibrary{decorations.pathreplacing,positioning,arrows.meta}
\usepackage{booktabs}
\usepackage{algorithm}
\usepackage{algpseudocode}
\usepackage{xcolor}
\usepackage{upgreek}
\usepackage{enumitem}
\usepackage{eufrak}
\usepackage{mathtools}
\usepackage{float}
\usepackage{lineno}

\newtheorem{thm}{Theorem}[section]
\newtheorem{corr}{Corollary}[section]
\newtheorem{lemma}{Lemma}[section]
\newtheorem{definition}{Definition}[section]

\begin{document}

\title[Bayesian Gaussian Mixture Modeling]{Bayesian Gaussian Mixture Modeling for Symmetric Matrix Variate Data}\thanks{CONTACT M. Wolff. Email: mlw32@uw.edu}

\author{Malcolm Wolff\textsuperscript{$1$}, Grace S. Chiu\textsuperscript{$4,1,3,5$}, Anton H. Westveld \textsuperscript{$2,3$}\\ and Adrian Dobra\textsuperscript{$1$}}
\address{\textsuperscript{$1$} Department of Statistics, University of Washington, Seattle, WA, USA; \textsuperscript{$2$} Research School of Finance and Actuarial Studies, Australian National University, Canberra, ACT, Australia; \textsuperscript{$3$} Department of Mathematics and Statistics, Virginia Commonwealth University, Richmond, VA, USA; \textsuperscript{$4$} William \& Mary’s Batten School of Coastal \& Marine Sciences, Virginia Institute of Marine Science, Gloucester Point, VA, USA; \textsuperscript{$5$} Department of Statistics and Actuarial Science, University of Waterloo, Waterloo, ON, USA}

\begin{abstract}
Statistical inference on individual activity networks has been a historically difficult task due to the lack of available data at the appropriate granularity and the complexity of modeling individual mobility patterns. The recent availability of GPS data from individual devices, combined with highly detailed demographic information, suggests that one of these challenges can now be addressed. We introduce a new model which we call the Symmetric Matrix-Variate Normal Mixture Model (\texttt{STRUCTURED}) to estimate how demographic traits influence changes in human activity networks, using sociomatrices that capture the probabilistic spatial overlap between individuals over time. We exploit the commutativity constraint inherent in the symmetric matrix-variate normal distribution to parameterize the column precision matrix as a polynomial of the row precision matrix, reducing the effective parameter space by an order of magnitude. We develop two variants of \texttt{STRUCTURED}: \texttt{STRUCTURED-FP}, which estimates the full polynomial, and \texttt{STRUCTURED-RJ}, which uses reversible-jump MCMC to select a reduced-order parameterization. Simulation studies demonstrate that \texttt{STRUCTURED-RJ} outperforms existing methods in sparse-data regimes, whereas \texttt{STRUCTURED-FP} is preferred when sample sizes are large. We apply the model to GPS-derived sociomatrices of 293 individuals in King County, WA, finding that local crime environments and youth employment density are the dominant demographic factors explaining variation in weekly activity overlap patterns.
\end{abstract}

\maketitle

\date{\today} 

\tableofcontents

\section{Introduction}
\label{sec:introduction}

Social and geographic networks are fundamentally intertwined. A range of classic studies demonstrate that social relationships in the United States form a “small world” network, where any two randomly chosen individuals can be connected socially \citep{travers1977experimental, korte1970acquaintance, killworth1978reversal} through an average of six intermediaries. Yet, geography continues to play a major role in shaping these networks: \citet{sorenson2003social} argue that industries agglomerate because entrepreneurs struggle to tap into resources that are geographically distant; \citet{onnela2011geographic} show that small social groups are typically highly localized in space and that social diffusion processes encounter clear structural and spatial limitations; and \citet{scellato2011socio} find that different social networking platforms display distinct geo-social patterns—services centered on location-based advertising predominantly reinforce local connections and clusters, whereas platforms oriented toward news and content sharing tend to foster more long-distance ties and clusters. 

Outside of social networks per se, sociodemographic and socioeconomic factors also strongly shape human mobility. \citet{barbosa2021uncovering} report that in certain Brazilian regions, mobility is closely tied to wealth, with low reliance on public transportation and communities segmented by income; \citet{gauvin2020gender} show that women frequent fewer distinct locations than men and allocate their time less evenly across them; and \citet{crutchfield1982crime} propose that high levels of mobility undermine social integration and, in turn, weaken informal community mechanisms for deterring crime.

The study of how social networks relate to human mobility has long been constrained by the lack of suitable data. Traditionally, mobility information has come from opt-in surveys and aggregated population measures. The proliferation of GPS-enabled personal devices now makes it possible to capture human movement patterns with unprecedented spatial and temporal detail. At the same time, high-resolution sociodemographic and socioeconomic datasets have become widely accessible. Together, these emerging data sources enable a more refined exploration of how social characteristics are linked to mobility.

Yet, only a limited set of methodological tools is currently available for analyzing high-resolution social and mobility data. GPS trajectories are collected for many actors over time and are typically converted into a series of sociomatrices by tabulating pairwise relationships among a fixed set of actors \citep{narayan2011peer, chung2005exploring, pramanik2016framework, rauch2022analysing, bapierre2011variable}, even though many actor pairs exhibit little to no geographic overlap. This yields a time-indexed collection of sparse, high-dimensional, symmetric sociomatrices, for which no state-of-the-art network estimation procedures are currently available. \citet{hoff2004modeling} introduce a family of latent space models to estimate relational effects in social networks, showing that their approach accommodates a wide range of latent distributions and produces visually interpretable results; however, several difficulties arise when this framework is applied to individual mobility data. The model does not incorporate the inherently dynamic temporal structure of human mobility networks, it does not capture relationships across multiple networks, and estimating each large network via vectorization becomes computationally infeasible in high dimensions. \citet{dobra2011bayesian} propose a scalable network model under a matrix-variate normal assumption, where correlations factor into a Kronecker product of row- and column-specific covariance matrices, thereby reducing the number of parameters by an order of magnitude; this decomposition, however, is not identifiable when the matrix-variate observations are symmetric. \citet{hanneke2009discrete} develop a class of discrete-time models for evolving networks, known as Temporal Exponential Random Graph Models (TERGMs), but these models treat the network as binary, whereas our data consist of continuous-valued sociomatrices.

The key contribution of this paper is \texttt{STRUCTURED} (\textbf{S}ymmetric ma\textbf{T}rix va\textbf{R}iate ga\textbf{U}ssian graphi\textbf{C}al mix\textbf{TURE} mo\textbf{D}el), a Bayesian mixture model for \textit{network-to-network} inference on symmetric relational data in sparse high dimensional settings. We expand on the previous theoretic development of symmetric matrix variate normal distributions \citep{nel1978symmetric, gupta2018matrix} by deriving conditional independence structures in the context of graphical models, and develop a novel methodology for joint sampling of commutative precision matrices when one is $G$-Wishart distributed. We leverage the spectral properties of symmetric matrices to sample from the posterior using a computationally efficient Metropolis-Hastings algorithm. As a use case, we use publicly available crime incidence data and census employment statistics to infer their relationship with overlap in human mobility patterns at the individual level, where sociomatrices of individual human mobility overlap are estimated using the GPS activity of 293 unique individuals in King County, WA from 11/05/2018 to 01/21/2019. Each sociomatrix summarizes weekly spatio-temporal overlap of individuals as a sparse symmetric matrix. 

The remainder of the paper proceeds as follows. In Section~\ref{sec:background}, we provide background on Bayesian matrix-variate inference, formally introducing Gaussian graphical models, the matrix-variate normal distribution, and the symmetric matrix-variate normal distribution. In Section~\ref{sec:methods}, we develop the methodological framework for Bayesian inference on symmetric matrix-variate distributions, describing the implications of symmetric matrix-variate data, introduce the \texttt{STRUCTURED} model and describe how sampling and posterior estimation are performed. In Section~\ref{sec:simulation}, we describe the behavior of the model using simulated data. In Section~\ref{sec:application}, we apply the model to symmetric matrix-variate data from human activity overlap in the real-world and infer the relationship between the overlap of observed human activity and local demographic graphs. Finally, in Section~\ref{sec:discussion}, we discuss implications and future directions.

\section{Bayesian Matrix Variate Theory}
\label{sec:background}

In this section, we present the necessary background on Bayesian matrix-variate inference. We begin by defining the matrix-variate normal and the symmetric matrix-variate normal distributions, emphasizing the specific properties that underpin our model construction. Next, we summarize key ideas from Bayesian matrix-variate Gaussian Graphical Models, focusing on Gaussian Graphical Models (GGMs) and the $G$-Wishart distribution. Lastly, we revisit Gaussian mixture models (GMMs).

\paragraph{The Matrix-Variate Normal Distribution} The matrix-variate normal distribution, which specifies a normal law for random matrices, is widely employed in the analysis of lattice-structured data. It models the covariance across rows and columns separately, thereby simplifying the overall covariance pattern among entries of the random matrix. We formally define this distribution below \citep{gupta2018matrix}.

\begin{definition}[Matrix Variate Normal]
The random matrix $\mathbf{X} \in \mathbb{R}^{p\times p}$ is said to have a matrix-variate normal distribution with mean matrix $\mathbf{M}$, row precision matrix $\mathbf{\Omega}$ and column precision matrix $\mathbf{\Upsilon}$--denoted as $\mathbf{X} \sim MN_{p,p}\left(\mathbf{M},\mathbf{\Omega},\mathbf{\Upsilon}\right)$--if it is distributed according to the probability density function
\begin{equation}
\begin{aligned}
f(\mathbf{X}\ |\ \mathbf{M}, \mathbf{\Omega}, \mathbf{\Upsilon}) &\equiv (2\pi)^{-\frac{1}{2}p^2}
\text{det}(\mathbf{\Omega})^{\frac{1}{2}p}
\text{det}(\mathbf{\Upsilon})^{\frac{1}{2}p}\\
&\phantom{\equiv}\times \text{etr}\left[-\dfrac{1}{2} \mathbf{\Omega}(\mathbf{X} - \mathbf{M})\mathbf{\Upsilon}(\mathbf{X} - \mathbf{M})^\top\right].
\end{aligned}
\label{eq:mvnormal}
\end{equation}
\label{def:mvnormal}
\end{definition}
This distribution may also be expressed in terms of the multivariate normal using column-wise vectorization of the matrix $\mathbf{X}$, \begin{equation}
\begin{aligned}
    \nonumber
	\text{vec}(\mathbf{X}) &\equiv (x_{11}, x_{21}, \ldots, x_{p1}, x_{12}, x_{22}, \ldots, x_{pp})^\top\\
\end{aligned}
\end{equation}
and the Kronecker product $\mathbf{\Omega}\otimes\mathbf{\Upsilon} = \begin{bmatrix}\omega_{ij}\mathbf{\Upsilon} \end{bmatrix} \in \mathbb{R}^{p^2\times p^2}$. In particular, $\mathbf{X}$ is distributed as a matrix variate normal with probability density function defined in Eq. \eqref{eq:mvnormal} and parameters $\mathbf{M}, \mathbf{\Omega}, \mathbf{\Upsilon}$ if and only if $\text{vec}(\mathbf{X})$ is distributed as a multivariate normal with parameters $\text{vec}(\mathbf{M}), \mathbf{\Omega}\otimes\mathbf{\Upsilon}$  \citep{gupta2018matrix}.

\paragraph{The Symmetric Matrix-Variate Normal Distribution} Two primary difficulties arise in extending this definition to symmetric random matrices $\mathbf{X} \in \mathbb{R}^{p\times p}$. First, even if $\mathbf{X}$ and $\mathbf{X}^\top$ are assumed to have a matrix variate normal distribution, $\mathbf{X} \neq_d \mathbf{X}^\top$ since the Kronecker product is non-commutative. Second, a $(p\times p)$ random matrix $\mathbf{X} = \mathbf{X}^\top$ has only $\frac{1}{2}p(p+1)$ unique random variables, and consequently $\mathbf{\Omega}\otimes\mathbf{\Upsilon}$ is singular.

The distributions of $\mathbf{X}$ and $\mathbf{X}^\top$ are equal when $\mathbf{\Omega}\otimes\mathbf{\Upsilon} = \mathbf{\Upsilon}\otimes\mathbf{\Omega}$. The commutativity of the Kronecker product can be achieved by requiring the commutativity of the individual matrices, as shown in the following lemmas that are proved in the Appendix.

\begin{lemma}
Suppose $\mathbf{\Omega\Upsilon} = \mathbf{\Upsilon\Omega}$. Then $\mathbf{\Omega}$ and $\mathbf{\Upsilon}$ are simultaneously diagonalizable.
\label{lemma:simul_diag}
\end{lemma}

\begin{lemma}
Let $\mathbf{\Omega}$, $\mathbf{\Upsilon} \in \mathbb{R}^{p\times p}$ be real positive symmetric definite matrices. Suppose $\mathbf{\Omega}\mathbf{\Upsilon} = \mathbf{\Upsilon}\mathbf{\Omega}$. Then $\mathbf{\Omega}\otimes\mathbf{\Upsilon} = \mathbf{\Upsilon}\otimes\mathbf{\Omega}$. 
\end{lemma}

In order to define a valid precision matrix, the symmetric matrix-variate normal distribution introduces the transformation matrix $\mathbf{B}_p \in \mathbb{R}^{p^2\times \tfrac{1}{2}p(p+1)}$ with elements
\begin{equation}
\begin{aligned}
(\mathbf{B}_p)_{ij, gh} &= \frac{1}{2}(\delta_{ig}\delta_{jh} + \delta_{ih}\delta_{jg}), & & i \leq p, j \leq p, g \leq h \leq p
\end{aligned}
\label{eq:Bp}
\end{equation}
where $\delta_{ij}$ denotes Kronecker's delta, taking the value 1 if $i = j$ and $0$ otherwise \citep{nel1978symmetric, gupta2018matrix}. This transformation matrix reduces the column-wise vectorization of a symmetric matrix $\mathbf{X}$ to the $\frac{1}{2}p(p+1)$ dimensional column-wise vectorization of its upper triangular and diagonal elements:
\begin{equation}
\begin{aligned}
	\mathbf{B}_p^\top\text{vec}(\mathbf{X}) &= (x_{11}, x_{12}, x_{22}, x_{13}, x_{23}, x_{33}, \ldots, x_{pp})^\top\\
	&\equiv \text{vecp}(\mathbf{X}) \in \mathbb{R}^{\tfrac{1}{2}p(p+1)}.
\end{aligned}
\end{equation}
Given two precision matrices $\mathbf{\Omega}, \mathbf{\Upsilon}$, a valid covariance matrix for symmetric matrix-variate data can be constructed as $$\mathbf{B}_p^\top(\mathbf{\Omega}^{-1}\otimes\mathbf{\Upsilon}^{-1})\mathbf{B_p} \in \mathbb{R}^{\tfrac{1}{2}p(p+1)\times\tfrac{1}{2}p(p+1)}.$$
One can then define the symmetric matrix-variate normal distribution as follows \citep{nel1978symmetric, gupta2018matrix}.

\begin{definition}[Symmetric Matrix Variate Normal]
The probability density function for a symmetric random matrix $\mathbf{X} \in \mathbb{R}^{p\times p}$ that follows the symmetric matrix variate normal distribution $\mathcal{SN}_{p, p}(\mathbf{M}, \mathbf{\Omega}, \mathbf{\Upsilon})$ has the form
\begin{equation}
\begin{aligned}
p(\mathbf{X}\ |\ \mathbf{M}, \mathbf{\Omega}, \mathbf{\Upsilon}) 
	&= (2\pi)^{-\frac{1}{4}(p+1)}
\text{det}(\mathbf{B}_p^\top(\mathbf{\Omega}^{-1}\otimes\mathbf{\Upsilon}^{-1})\mathbf{B}_p)^{-\frac{1}{2}}\\
&\phantom{\equiv}\times \text{etr}\left[ -\tfrac{1}{2} \mathbf{\Omega}(\mathbf{X} - \mathbf{M})\mathbf{\Upsilon}(\mathbf{X} - \mathbf{M}) \right]
\end{aligned}
\label{eq:smvn}
\end{equation}
where $\mathbf{M},\ \mathbf{\Omega},\ \mathbf{\Upsilon}$ are constant symmetric $p \times p$ matrices such that $\mathbf{\Omega}\mathbf{\Upsilon} = \mathbf{\Upsilon}\mathbf{\Omega}$, and $\mathbf{B}_p$ is the transition matrix in equation \eqref{eq:Bp}.
\label{def:smvn}
\end{definition}

Studied by \citet{nel1978symmetric}, the symmetric matrix-variate normal distribution has been used for local approximation to Wishart random matrices \citep{ouimet2022symmetric}, and its characteristic function and probability density function appear as the limiting distribution of a Wishart variate matrix \citep{anderson1958absence}.

\paragraph{Bayesian Inference for Gaussian Graphical Models}
Gaussian graphical models (GGMs) \citep{dempster1972covariance} leverage the covariance structure of multivariate Gaussian data, where the appearance of zeros in the precision matrix implies conditional independence. For a random $p$-dimensional vector $\underline{X} \sim \mathcal{N}(\underline{0}, \mathbf{\Omega}^{-1})$, a graph $G \equiv (V, E)$ for vertex set $V \equiv \{1,\ldots,p\}$ and edge set $E \subset \{(i,j)\ |\ 1\leq i < j \leq p\}$, the vector $\underline{x}$ is represented by the Gaussian graphical model with conditional independence graph $G$ if 
$$x_i\ \bot\ x_j\ |\ \{x_k,\ k\notin (i,j)\} \text{ whenever } (i,j) \notin E.$$
That is, $x_i$ and $x_j$ are conditionally independent when the dyad $(i,j)$ is not an element of the edge set $E$.

Various Bayesian approaches have been proposed for estimation and inference on GGMs by imposing a prior distribution on $\mathbf{\Omega}$ \citep[e.g.,][]{leonard1992bayesian, yang1994estimation, barnard2000modeling, smith2002parsimonious, liechty2004bayesian, rajaratnam2008flexible}, and selecting the maximum a posteriori (MAP) estimates or using Bayesian model averaging to account for the posterior uncertainty \citep{kass1995bayes}. \citet{roverato2002hyper} derived a conjugate prior for $\mathbf{\Omega}$ with support induced by the graph $G$, known as the $G-$Wishart distribution, defined below.

\begin{definition}[The $G$-Wishart Distribution]
The probability density function for a symmetric positive definite random matrix $\mathbf{\Omega} \in \mathbb{R}^{p\times p}$ with support described by the graph $G$ which follows the $G-$Wishart distribution $\text{Wis}_G(\delta, \mathbf{D})$ has the form
\begin{equation}
    \begin{aligned}
        p(\mathbf{\Omega}|G, \delta, \mathbf{D}) &= \dfrac{|\mathbf{\Omega}|^{(\delta-2)/2}}{I_G(\delta, \mathbf{D})}\exp\biggl\{-\frac{1}{2}\text{tr}\left[\mathbf{\Omega}, \mathbf{D}\right]\biggr\},
    \end{aligned}
\end{equation}
where $I_G(\delta, \mathbf{D})$ is a normalizing constant.
\end{definition}

The probability density function of the characteristic roots {\small $\lambda = \lambda_1 \geq \ldots \geq \lambda_p > 0$} for a symmetric positive definite random matrix $\mathbf{K} \in \mathbb{R}^{p\times p}$ that follows the Wishart distribution $\text{Wis}(\delta, \mathbf{I})$ has the form
\begin{equation}
\begin{aligned}
p(\underline{\lambda} | \delta, \mathbf{I})
	&:= \dfrac{1}{I(p, \delta)}
	\prod_{i=1}^p \lambda_i^{\frac{1}{2}(\delta - 2)}
	\exp{\left[ -\frac{1}{2}\sum_{i=1}^p \lambda_i \right]}
	\prod_{i < j} (\lambda_i - \lambda_j),
\end{aligned}
\label{eq:eigen}
\end{equation}
and is distributed independently of the eigenvectors, in the range where the density is not zero \citep{anderson2003textit}. In particular, Eq. \eqref{eq:eigen} suggests that a Wishart distributed matrix has distinct roots with probability $1$.

\paragraph{Gaussian Mixture Models} An $n \times p$ data matrix $\mathbf{X}$ that follows a finite mixture of $p$-dimensional Gaussian distributions with the mixing parameter vector $\underline{\pi}$ and the mean and precision parameters $\mathbf{\Theta} \equiv \{\underline{\mu}_l, \mathbf{\Omega}_l\}_{l=1}^L$ has the density
\begin{equation}
    \begin{aligned}
        \mathbf{X}\ |\ \underline{\pi}, L, \mathbf{\Theta} &\sim  \prod_{i=1}^N \biggl[\sum_{l=1}^L \pi_l\mathcal{N}(\underline{x}_i\ ;\ \underline{\mu}_l,\ \mathbf{\Omega}_l^{-1}) \biggr],
    \end{aligned}
    \label{eq:mixture1}
\end{equation}
where $\sum_{l=1}^L \pi_l = 1$.
The practical application of the estimation of the mixture model uses data augmentation to introduce labels $\{c_i\}_{i=1}^N$ representing the membership of the cluster component of each observation $i$, changing the data density in Eq. \eqref{eq:mixture1} to
\begin{equation}
    \begin{aligned}
        \mathbf{X}, \underline{c}\ |\ \underline{\pi}, L, \mathbf{\Theta}  &\sim  \prod_{l=1}^L \prod_{i: c_i = l} \pi_l\mathcal{N}(\underline{x}_i\ ;\ \underline{\mu}_l,\  \mathbf{\Omega}_l^{-1}),
    \end{aligned}
    \label{eq:mixture2}
\end{equation}
where the inner product is over all observations $i$ in group $l$.  Due to their ease of estimation and well-behaved properties, Gaussian mixture models continue to enjoy a large amount of use in the network estimation literature \citep{hoff2002latent, ryan2017bayesian}.

\section{Methods}
\label{sec:methods}

In this section, we present the methodological development of \texttt{STRUCTURED}. We begin by deriving consequences of the commutativity assumption, which in turn motivate a reduced-order parametrization of $\mathbf{\Upsilon}$. Next, we introduce the symmetric matrix-variate Gaussian mixture model and outline the posterior sampling scheme. We consider two variants: \texttt{STRUCTURED-FP}, which estimates the full polynomial $\mathbf{\Upsilon} = \sum_{k=0}^{p-1}\alpha_k\mathbf{\Omega}^k$; and \texttt{STRUCTURED-RJ}, which restricts attention to a single active term of the form $\mathbf{\Upsilon} = a\mathbf{I} + b\mathbf{\Omega}^m$, with the order $m$ determined via reversible-jump MCMC. The appropriate choice between these variants depends on the effective sample size, as demonstrated in a simulation study.

\paragraph{Implications of Commutativity for $G$-Wishart Matrices} The commutativity of the precision matrices in Definition \ref{def:smvn} has important implications for their joint prior specification. We derive a precise relationship between the distribution of $\mathbf{\Omega}$ and $\mathbf{\Upsilon}$ and its implication when $\mathbf{\Omega} \sim \text{Wis}_G(\delta, \mathbf{I})$. 

First, we show  that if two deterministic matrices $\mathbf{\Omega}, \mathbf{\Upsilon}$ commute and $\mathbf{\Omega}$ has only simple roots, $\mathbf{\Upsilon}$ may be written as a polynomial in $\mathbf{\Omega}$.

\begin{lemma}
Suppose $\mathbf{\Omega}, \mathbf{\Upsilon}$ are simultaneously diagonalizable, and $\mathbf{\Omega}$ has only simple roots. Then $\mathbf{\Upsilon} = \sum_{k=0}^{p-1} \alpha_k\mathbf{\Omega}^k$.
\label{lemma:poly}
\end{lemma}

Theorem \ref{thm:poly2} specifies the structure of the conditional distribution of $\mathbf{\Upsilon} \mid \mathbf{\Omega}$ for $\mathbf{\Omega} \sim \text{Wis}_G(\delta, \mathbf{I})$.

\begin{thm}
Suppose $\mathbf{\Upsilon}$ commutes with $\mathbf{\Omega}$. Then 
$$\mathbf{\Upsilon} = \lim_{\varepsilon \to 0} \sum_{k=0}^{p-1} \alpha_k(\mathbf{\Omega} + \varepsilon)^k$$
for some $\alpha_0, \ldots, \alpha_{p-1} \in \mathbb{R}$. If $\mathbf{\Omega}$ has only simple roots, then one can take $\varepsilon = 0$.
\label{thm:poly2}
\end{thm}

Since $\mathbf{\Upsilon}$ is arbitrarily close to a polynomial of $\mathbf{\Omega}$, the associated sparsity structure of $\mathbf{\Upsilon}$ is governed by the sparsity structure of powers of $\mathbf{\Omega}$. We state this in Corollary \ref{corr:graphs}.
\begin{corr}
Let the graph $G_{\mathbf{\Omega}}$ represent the sparsity structure of $\mathbf{\Omega}$, and let $\mathbf{A}_{\mathbf{\Omega}}$ be the adjacency matrix associated with $G_{\mathbf{\Omega}}$. Then
\begin{equation}
    \nonumber
    \begin{aligned}
        \mathbf{A}_{\mathbf{\Upsilon}} \in \{\mathbf{A}_{\mathbf{\Omega}}, \mathbf{A}_{\mathbf{\Omega^2}},\ldots, \mathbf{A}_{\mathbf{\Omega}^{d-1}}\}.
    \end{aligned}
\end{equation}
\label{corr:graphs}
\end{corr}
Moreover, the conditional independence structure of the matrix-variate data $\mathbf{X}$ is given by an explicit function of the sparsity of $\mathbf{\Omega}$ and $\mathbf{\Upsilon}$. Specifically, if $\mathbf{X}$ follows a symmetric matrix-variate normal distribution, the partial correlations associated with $\mathbf{X}$ are $0$ if and only if
$$\sigma_{ik}\psi_{jl} + \sigma_{jk}\psi_{il} + \sigma_{il}\psi_{jk} + \sigma_{jl}\psi_{ik} = 0,$$
where $\mathbf{\Omega}^{-1} = \mathbf{\Sigma} = (\sigma_{ij})$ and $\mathbf{\Upsilon}^{-1} = \mathbf{\Psi} = (\psi_{ij})$.

\paragraph{Reduced-Order Parameterization}
\label{sec:rjmcmc}
While Theorem \ref{thm:poly2} establishes that $\mathbf{\Upsilon}$ is a degree $p-1$ polynomial of $\mathbf{\Omega}$, estimating all $p$ coefficients $\alpha_0,\ldots,\alpha_{p-1}$ poses challenges: the coefficient vector is non-identifiable up to a rescaling of $\mathbf{\Omega}$ (see \ref{appendix:spherical}), and with $N$ observations of $p \times p$ matrices, $p$ additional parameters are difficult to estimate reliably.

We address both issues by restricting to a single active polynomial term:
\begin{equation}
    \mathbf{\Upsilon} = a\mathbf{I} + b\mathbf{\Omega}^m,
    \label{eq:reduced_order}
\end{equation}
where $a \geq 0$, $b > 0$, and the polynomial order $m \in \{1,\ldots,K\}$ is treated as an unknown parameter. This reduces the parameterization from $p$ free coefficients to 3 effective degrees of freedom ($a$, $b$, $m$). Scale identifiability is resolved by normalizing $\mathbf{\Omega}_{1,1} = 1$ via Cholesky decomposition \citep{dobra2011bayesian}. The identity component $a\mathbf{I}$ absorbs residual variation not explained by $\mathbf{\Omega}^m$, improving robustness to misspecification of $m$.

\paragraph{Symmetric Matrix-Variate Gaussian Mixture Model} We extend the model to a collection of $N$ symmetric matrices $\mathcal{D} \equiv \{\mathbf{X}_i\}_{i=1}^N$ and parameters $\mathbf{\Theta} \equiv \{\mathbf{M}_l, \mathbf{\Omega}_l, \mathbf{\Upsilon}_l\}_{l=1}^L$ as
\begin{equation}
    \begin{aligned}
        \mathcal{D}, \underline{c}\ |\ \underline{\pi}, L, \mathbf{\Theta}  &\sim  \prod_{l=1}^L \prod_{i: c_i = l} \pi_l\mathcal{SN}(\mathbf{X}_i\ ;\ \mathbf{M}_l,\  \mathbf{\Omega}_l,\ a_l\mathbf{I} + b_l\mathbf{\Omega}_l^{m_l}).
    \end{aligned}
    \label{eq:mixture3}
\end{equation}
By Corollary \ref{corr:graphs}, a mixture model with a representative collection of $L$ graphs $\{G_l\}_{l=1}^L$ jointly explores the correlation structure of both $\mathbf{\Omega}$ and $\mathbf{\Upsilon}$.

\paragraph{Posterior Estimation.} We assume knowledge of the number of groups $L$ and the associated graphs $\{G_l\}_{l=1}^L$. The full posterior is
\begin{equation}
    \begin{aligned}
        p(\mathbf{\Theta}, \underline{c} \ |\ \mathcal{D}, L) 
        & \propto
        \prod_{l=1}^L \left[\prod_{i : c_i = l} \pi_l\mathcal{SN}(\mathbf{X}_i\ ; \ \mathbf{\Theta}_l)\right]
        \times\\
        & \prod_{l=1}^L \biggl[ p(\mathbf{M}_l) p(\mathbf{\Omega}_l) p(a_l) p(b_l) p(m_l) \biggr]
         p(\underline{c} | L)p(\underline{\pi} | L),
    \end{aligned}
    \label{eq:fullposterior}
\end{equation}
with priors: $\mathbf{\Omega}_l \sim \text{Wish}_{G_l}(\delta_l, \mathbf{I})$, $p(m_l) \propto \rho^{m_l}$ (geometric, favoring low order), $a_l \sim \text{Exp}(\lambda_a)$, $b_l \sim \text{Exp}(\lambda_b)$, $\underline{\pi} \sim \text{Dir}(1,\ldots,1)$.

We sample from the posterior using the following Metropolis-Hastings moves within each iteration:
\begin{enumerate}
    \item \textit{Cluster assignments $\underline{c}$}: Gibbs update, sampling $c_i$ from conditional multinomial proportional to $\pi_l \mathcal{SN}(\mathbf{X}_i; \mathbf{\Theta}_l)$.
    \item \textit{Mixing weights $\underline{\pi}$}: Gibbs update from the Dirichlet full conditional.
    \item \textit{Precision matrix $\mathbf{\Omega}_l$}: Cholesky perturbation proposal \citep{dobra2011bayesian}. Elements $(\Psi_l)_{ij}$ are perturbed when $\{i,j\} \in E_l$ or $i=j$; non-free elements are completed to satisfy the graph constraint. We normalize $\mathbf{\Omega}_{l,11} = 1$ after each proposal.
    \item \textit{Scale parameters $(a_l, b_l)$}: Random walk on $a_l$ (reflecting at 0) and log-scale random walk on $b_l$.
    \item \textit{Polynomial order $m_l$}: Reversible-jump move. Propose $m_l' = m_l \pm 1$ with equal probability (reflecting at boundaries), generate $b_l' \sim \text{Exp}(1)$, and accept with probability
\begin{equation}
    \alpha = \min\Biggl\{1,\ \dfrac{\mathcal{L}(\mathcal{D}\ |\ \mathbf{\Omega}_l, a_l, b_l', m_l')}{\mathcal{L}(\mathcal{D}\ |\ \mathbf{\Omega}_l, a_l, b_l, m_l)} \cdot \dfrac{P(m_l')}{P(m_l)} \cdot \dfrac{q(b_l\ |\ m_l)}{q(b_l'\ |\ m_l')} \Biggr\}.
\end{equation}
\end{enumerate}

The data likelihood is efficiently computed via eigendecomposition $\mathbf{\Omega}_l = \mathbf{U}_l\mathbf{\Lambda}_l\mathbf{U}_l^\top$, producing $\mathbf{\Upsilon}_l = \mathbf{U}_l\text{diag}(a_l + b_l\lambda_i^{m_l})\mathbf{U}_l^\top$ and the normalizing constant
\begin{equation}
    \text{det}\left[\mathbf{B}_p^\top(\mathbf{\Omega}_l^{-1}\otimes\mathbf{\Upsilon}_l^{-1})\mathbf{B}_p\right]^{-1/2}
    \propto \prod_{i \leq j} \left(\sigma_i\psi_j + \sigma_j\psi_i\right)^{-1/2},
\end{equation}
where $\sigma_i = \lambda_i^{-1}$ and $\psi_i = (a_l + b_l\lambda_i^{m_l})^{-1}$, avoiding materialization of $\mathbf{B}_p$.

\subsection{Multimodality}

The posterior distribution is naturally multimodal due to permutation symmetries in the eigenvalues of $\mathbf{\Omega}$ and $\mathbf{\Upsilon}$. To improve mixing, we augment the sampler with parallel tempering \citep{sambridge2014parallel}, running $M$ chains at temperatures $T_1 = 1 < T_2 < \ldots < T_M$ and proposing swaps between adjacent chains with probability
\begin{equation}
    \min\biggl\{1, \left(\dfrac{p(\mathbf{\Theta}_i\ |\ \mathcal{D})}{p(\mathbf{\Theta}_j\ |\ \mathcal{D})}\right)^{T_j^{-1} - T_i^{-1}}\biggr\}.
\end{equation}
The cold chain ($T_1 = 1$) targets the true posterior and is used for inference; the heated chains explore more broadly and help escape local modes.

\section{Simulation Study}
\label{sec:simulation}

\paragraph{Experimental Setup} To assess and compare the efficacy of the model for symmetric matrix data, we perform simulations with the specification following \citet{dobra2011bayesian}. We consider the cycle graph $C_p \in \mathcal{G}_p$ with edges $\{(i,i+1):1\leq i \leq p - 1\} \cup \{(p,1)\}$ and the associated precision matrix $\mathbf{A} \in \mathbf{P}_{C_p}$ to be $1$ on the diagonal elements, $.5$ on the off-diagonal elements, $\mathbf{A}_{p,1} = \mathbf{A}_{1,p} = .4$. For $p = 5$, we generate $N=20$ data observations $\mathbf{X}_i \sim \mathcal{SN}(\mathbf{0}, \mathbf{A}, \mathbf{A}^{m})$ for $m \in \{1, 2\}$. The sampling is initialized with a random draw $\mathbf{\Omega}^{(0)} \sim \text{Wis}_{G_\mathbf{\Omega}}(3, \mathbf{I}_{p})$, normalized so that $\mathbf{\Omega}^{(0)}_{1,1} = 1$ \citep{dobra2011bayesian}. The graph $G_\mathbf{\Omega} = \mathcal{C}_p$ is assumed to be known. We used 50,000 post-burn-in iterations with results averaged over 3 independent replications.

\paragraph{Convergence} Figure \ref{fig:elem_convergence} displays the differences between the estimated and true values of the aggregated matrix elements for $\mathbf{\Omega}$ and $\mathbf{\Upsilon}$ in 10 independent chains run for 100,000 iterations with $N=20, p = 5, m=1$. Although each chain starts significantly far apart, the cumulative averages become similar after 1,000 iterations and approximately converge after 100,000 iterations.

\begin{figure}[ht]
    \centering
    \includegraphics[scale=.4]{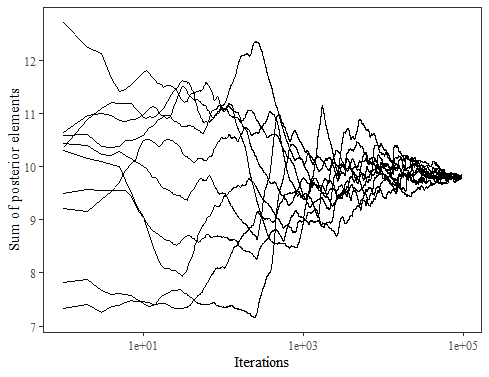}\hspace{.2in}\includegraphics[scale=.4]{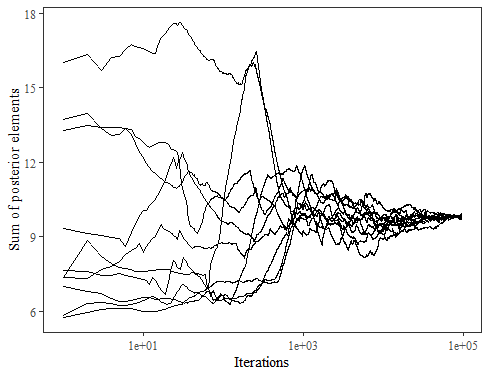}
    \caption{Cumulative average sum of matrix elements for $\mathbf{\Omega}$ (left panel) and $\mathbf{\Upsilon}$ (right panel) by iteration across 10 independent chains.}
    \label{fig:elem_convergence}
\end{figure}

\paragraph{Results} We compare \texttt{STRUCTURED-RJ} with the matrix-variate Gaussian graphical model (\texttt{MVGGM}) of \citet{dobra2011bayesian}, which estimates $\mathbf{\Omega}$ and $\mathbf{\Upsilon}$ independently without imposing commutativity. We assume that the graph $G_{\mathbf{\Omega}}$ is known for both methods and the graph $G_{\mathbf{\Upsilon}}$ is known for \texttt{MVGGM}.

\begin{table}[H]
    \centering
    \begin{tabular}{llccc}
     &  &  \multicolumn{3}{c}{{\small \textbf{MSEs}}} \\
    \cmidrule{3-5}
    &  {\small \textbf{Model}} & {\small ${\mathbf{\Omega}}$} & {\small ${\mathbf{\Upsilon}}$} & {\footnotesize  ${\mathbf{B}_p^\top(\mathbf{\Omega}\otimes\mathbf{\Upsilon})\mathbf{B}_p}$} \\ 
    \midrule
 $\mathbf{\Upsilon} = \mathbf{\Omega}$ & \texttt{STRUCTURED-RJ} & \textbf{.008} & \textbf{.019} & \textbf{.006}\\
 & \texttt{MVGGM} & .029 & .068 & .057\\
\cmidrule{2-5}
$\mathbf{\Upsilon} = \mathbf{\Omega}^2$ & \texttt{STRUCTURED-RJ} & \textbf{.010} & .162 & .047\\
 & \texttt{MVGGM} & .033 & \textbf{.096} & \textbf{.029}\\
    \bottomrule
    \end{tabular}
    \caption{Simulation results comparing \texttt{STRUCTURED-RJ} and \texttt{MVGGM} for $p=5$, $N=20$. The lowest MSEs are shown in boldface.}
    \label{table:main_results}
\end{table}

Table \ref{table:main_results} displays the mean squared errors (MSEs) of $\mathbf{\Omega}$, $\mathbf{\Upsilon}$, and the joint precision $\mathbf{B}_p^\top(\mathbf{\Omega}\otimes\mathbf{\Upsilon})\mathbf{B}_p$ for both specifications of $\mathbf{\Upsilon}$. When $\mathbf{\Upsilon} = \mathbf{\Omega}$, \texttt{STRUCTURED-RJ} achieves a substantially lower MSE in all three quantities, with $\text{MSE}(\mathbf{B}_p^\top(\mathbf{\Omega}\otimes\mathbf{\Upsilon})\mathbf{B}_p)$ roughly 10 times lower than \texttt{MVGGM}. The reversible-jump sampler correctly identifies $m=1$ with posterior probability $P(m=1) \approx 0.84$. When $\mathbf{\Upsilon} = \mathbf{\Omega}^2$, \texttt{STRUCTURED-RJ} achieves the lowest $\text{MSE}(\mathbf{\Omega})$ but exhibits higher $\text{MSE}(\mathbf{\Upsilon})$ than \texttt{MVGGM}. This difference is driven by statistical error amplification: if $\hat{\mathbf{\Omega}}$ has estimation error $\varepsilon$, then $\hat{\mathbf{\Upsilon}} = b\hat{\mathbf{\Omega}}^m$ inherits error of order $m\lambda_{\max}^{m-1}\varepsilon$, where $\lambda_{\max}$ is the largest eigenvalue of $\mathbf{\Omega}$. This amplification is inherent to the model -- as direct eigenvalue evaluation $\delta_i = b\lambda_i^m$ produces an identical error to forming the matrix power $b\mathbf{\Omega}^m$---and is unavoidable for any method that routes $\mathbf{\Upsilon}$ through powers of $\mathbf{\Omega}$.

\paragraph{Sensitivity to Parameterization} We consider two variants of \texttt{STRUCTURED} that exploit the polynomial structure of Theorem \ref{thm:poly2}: \texttt{STRUCTURED-RJ}, which uses the reduced-order parameterization $\mathbf{\Upsilon} = a\mathbf{I} + b\mathbf{\Omega}^m$ with reversible-jump on $m$; and \texttt{STRUCTURED-FP}, which estimates the full $p$-coefficient polynomial $\mathbf{\Upsilon} = \sum_{k=0}^{p-1}\alpha_k\mathbf{\Omega}^k$. We compare both against an oracle model (\texttt{FIXED}: same as \texttt{STRUCTURED-RJ} but with $m$ known), a reversible-jump variant on the truncation order of the full polynomial (\texttt{SUM-K}: $\mathbf{\Upsilon} = \sum_{k=0}^{K}\alpha_k\mathbf{\Omega}^k$ with RJ on $K$), and the baseline \texttt{MVGGM}.

Table \ref{table:full_poly} reports the MSEs of the joint precision $\mathbf{B}_p^\top(\mathbf{\Omega}\otimes\mathbf{\Upsilon})\mathbf{B}_p$ in three data-generating specifications, including a mixed-polynomial case where the truth has two active terms ($\mathbf{\Upsilon} = 0.5\mathbf{\Omega} + 0.5\mathbf{\Omega}^2$).

\begin{table}[H]
    \centering
    \begin{tabular}{llccc}
     &  &  \multicolumn{3}{c}{{\small \textbf{MSE}$(\mathbf{B}_p^\top(\mathbf{\Omega}\otimes\mathbf{\Upsilon})\mathbf{B}_p)$}} \\
    \cmidrule{3-5}
    &  {\small \textbf{Model}} & $m=1$ & $m=2$ & $0.5\mathbf{\Omega} + 0.5\mathbf{\Omega}^2$\\ 
    \midrule
 & \texttt{STRUCTURED-RJ} & \textbf{.006} & .047 & \textbf{.011}\\
 & \texttt{FIXED} (oracle) & \textbf{.006} & {.041} & .015\\
 & \texttt{STRUCTURED-FP} & .013 & \textbf{.025} & .020\\
 & \texttt{SUM-K} (RJ on $K$) & .020 & .146 & .043\\
 & \texttt{MVGGM} & .057 & .029 & .037\\
    \bottomrule
    \end{tabular}
    \caption{MSE of the joint precision across parameterization choices for $p=5$, $N=20$. The lowest MSEs are shown in boldface.}
    \label{table:full_poly}
\end{table}

We find that \texttt{STRUCTURED-RJ} performs within 5\% of the oracle \texttt{FIXED} for $m=1$ and achieves the lowest MSE for the mixed-polynomial case, confirming that the reversible-jump mechanism incurs negligible estimation cost. The \texttt{STRUCTURED-FP} model achieves the lowest MSE for $m=2$ at $N=20$, as its additional coefficients can compensate for the estimation error in $\hat{\mathbf{\Omega}}$; however, this advantage is specific to the moderate-$N$ regime. The \texttt{SUM-K} variant performs poorly across all specifications despite correctly identifying the polynomial structure: the additional degrees of freedom degrade estimation when $N=20$ observations are insufficient to support multiple active coefficients simultaneously.

\paragraph{Sample Size Sensitivity} The relative advantage of the reduced-order parameterization depends critically on the ratio of observations to parameters. Table \ref{table:sample_size} presents MSEs of the joint precision across sample sizes for the more challenging $m=2$ specification.

\begin{table}[H]
    \centering
    \begin{tabular}{lcccc}
    & \multicolumn{4}{c}{\small \textbf{MSE}$(\mathbf{B}_p^\top(\mathbf{\Omega}\otimes\mathbf{\Upsilon})\mathbf{B}_p)$ by $N$}\\
    \cmidrule{2-5}
    {\small \textbf{Model}} & 5 & 10 & 20 & 50 \\
    \midrule
    \texttt{STRUCTURED-RJ}  & \textbf{.044} & .045 & .047 & .052 \\
    \texttt{STRUCTURED-FP}   & .068 & \textbf{.039} & \textbf{.025} & \textbf{.027} \\
    \texttt{SUM-K}       & .130 & .147 & .146 & .146 \\
    \texttt{MVGGM}       & .145 & \textbf{.039} & .029 & .040 \\
    \bottomrule
    \end{tabular}
    \caption{MSE of $\mathbf{B}_p^\top(\mathbf{\Omega}\otimes\mathbf{\Upsilon})\mathbf{B}_p$ across sample sizes for $\mathbf{\Upsilon} = \mathbf{\Omega}^2$, $p=5$. The lowest MSEs are shown in boldface.}
    \label{table:sample_size}
\end{table}

We find that for $N = 5$, \texttt{STRUCTURED-RJ} achieves the lowest MSE, outperforming all alternatives by at least 1.5-fold. The crossover to \texttt{STRUCTURED-FP} and \texttt{MVGGM} occurs at $N \approx 10$, beyond which the more flexible models accumulate sufficient data to overcome their parametric disadvantage. The \texttt{SUM-K} variant performs poorly at all sample sizes (MSE $\geq$ 0.13), confirming that estimating multiple polynomial coefficients simultaneously is not viable in this setting. This sample size sensitivity is directly relevant to the application in Section \ref{sec:application}: with $L=3$ mixture components and $N=11$ weekly sociomatrices, each component is effectively estimated from $N/L \approx 4$ observations, placing the problem within the regime where \texttt{STRUCTURED-RJ} provides consistent gains. We therefore adopt \texttt{STRUCTURED-RJ} for the application, noting that \texttt{STRUCTURED-FP} would be preferred in data-rich settings ($N \geq 10$ per component) where the additional polynomial flexibility can be reliably estimated.

\section{Application: Human Activity Inference}
\label{sec:application}

We construct sociomatrices of human activity overlap estimated using the GPS activity of 293 unique individuals in King County, WA from 11/05/2018 to 01/21/2019. Each sociomatrix summarizes the spatiotemporal overlap of individuals over the course of a week as a sparse symmetric matrix. We conducted inference on the effects of individuals' demographic networks on changes in individual movement patterns during this time period using publicly available crime incidence data and census employment statistics.

\paragraph{Sociomatrix Construction} The data used to construct sociomatrices of human activity overlap contain recorded locations vended by SafeGraph for 293 mobile devices registered in a large metropolitan area over a 12 week period from 11/05/2018 to 01/21/2019. Each of the records have a unique device identifier, timestamped geographic coordinates measured in degrees latitude and longitude.

We construct sociomatrices using a L\'{e}vy Flight Cluster Model (LFCM) \citep{wolff2024levy}. For each individual $i$ and week $w$, we fit the LFCM to the observed GPS traces to obtain posterior estimates of the activity regions $\{(\underline{\mu}_g^{(i)}, \mathbf{\Sigma}_g^{(i)})\}_{g=1}^{G_i}$ and the movement parameters. The LFCM is a Bayesian memory-aware L\'{e}vy flight mixture model that represents individual mobility as a set of Gaussian activity regions connected by heavy-tailed flights, producing probabilistic estimates of where individuals spend time. We then generate continuous sample paths by Brownian bridge interpolation within activity regions and linear interpolation during flights, evaluated on a common temporal grid $t_1,\ldots,t_T$ spanning the entire week. The sociomatrix for week $w$ is constructed as
\begin{equation}
    \begin{aligned}
        s_{i,j}^{(w)} &= \frac{1}{S}\sum_{s=1}^S \left\| \mathbf{x}_i^{(s)}(\cdot) - \mathbf{x}_j^{(s)}(\cdot) \right\|_{L_2}^{-1},
    \end{aligned}
    \label{eq:socio_lfcm}
\end{equation}
where $\mathbf{x}_i^{(s)}(\cdot)$ denotes the $s$-th generated sample path for individual $i$ and the average is taken over $S=100$ posterior scans. This approach has two advantages over the conservative proportional time (CPT) estimator of \citet{dong2020statistical}: it yields sociomatrices with continuous support compatible with the Gaussian likelihood (3\% near-zero entries versus 57\% exact zeros for CPT), and it produces temporally stable estimates (consecutive-week rank correlation of 0.70 versus 0.38 for CPT).

Figure \ref{fig:example_activity} displays the estimated activity distribution for an example individual in a given week; latitude and longitude are normalized to preserve privacy. Although a fair amount of spread is observed across the space, the individual spends the majority of time in the week around (-1,1). 
\begin{figure}[ht]
    \centering
    \hspace{-.1in}\includegraphics[scale=.9]{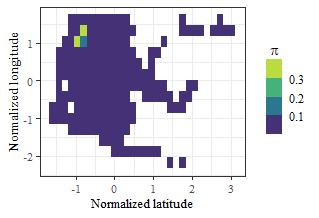}
    \caption{Example activity distribution estimated using the LFCM.}
    \label{fig:example_activity}
\end{figure}

Figure \ref{fig:log_distances} displays the distributions of pairwise log distances between individuals by week, where the distances are computed using Eq. \eqref{eq:socio_lfcm}. Colors indicate the number of missing devices in that particular week of the 293 total devices. Fifty devices were not observed in the first week, a substantially higher number than in the remaining weeks, ranging from 0--25 missing devices. Moreover, the observed pairwise distance distribution in week 1 admits a characteristically distinct shape. Due to these differences, the first week of data is removed from the remainder of analysis.
\begin{figure}[ht]
    \centering
    \hspace{.5in}\includegraphics[scale=.8]{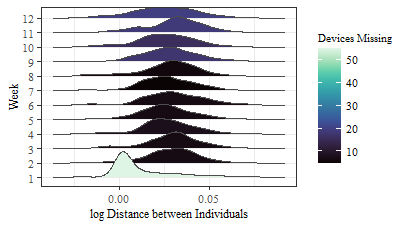}
    \caption{Log distances between individuals over 12 weeks. Color indicates how many of the 293 devices were missing during that week.}
    \label{fig:log_distances}
\end{figure}

Figure \ref{fig:sociomatrix_example} displays the resulting sociomatrix for a representative week, where entry $(i,j)$ represents the average inverse $L_2$ distance between the generated trajectories of actors $i$ and $j$.
\begin{figure}[ht]
    \centering
    \includegraphics[scale=.7]{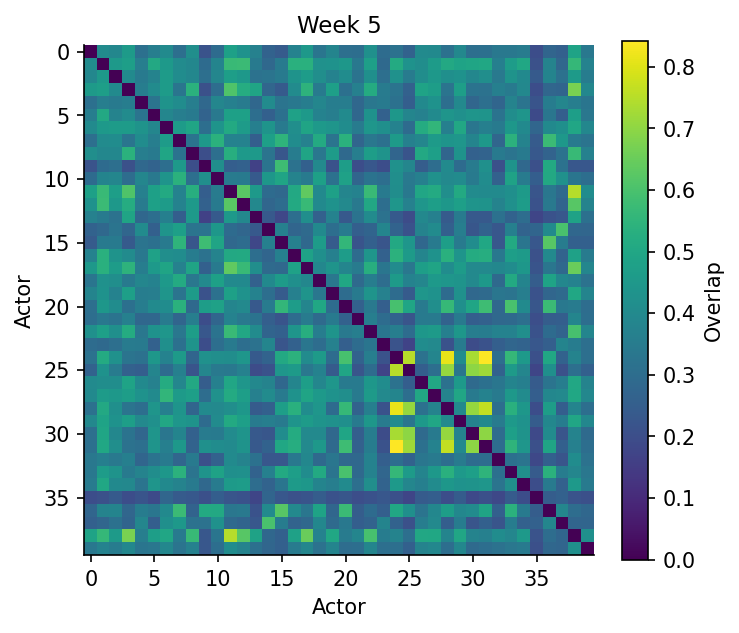}
    \caption{Example LFCM-based sociomatrix for week 5. Brighter values indicate greater spatiotemporal overlap between actors.}
    \label{fig:sociomatrix_example}
\end{figure}

\paragraph{Construction of Demographic Networks}

\par\noindent\textit{Historic Crime Incidence Reports} We construct a similarity graph using publicly available high-resolution spatio-temporal crime incidence data from the King County Sheriff's Office (KCSO) historic incident dataset. The data set contains 17,718 geocoded incidents with latitude, longitude, and type of incident. For each actor $i$, we compute the total number of crime incidents within a 5km radius of the actor's estimated home location, producing an exposure rate $\mathbf{R}_{\text{crime}}$. The graph $G_{\text{crime}}$ is then constructed by connecting actors $i$ and $j$ when their exposure rates are within a threshold: $|R_{\text{crime},i} - R_{\text{crime},j}| < \tau$, where $\tau$ is set at the 15th percentile of all pairwise differences. The resulting graph has 174 edges.\\

We construct two additional similarity graphs using publicly available census data from the Longitudinal Employer-Household Dynamics (LEHD) Origin-Destination Employment Statistics, which provide census block-level employment counts by age and earnings for 2018. For King County, the data set contains 21,229 census blocks with employment characteristics. We computed two exposure rates within 5km of each actor's home:
\begin{itemize}
    \item \textit{Youth employment} ($G_{\text{youth}}$): total number of workers age 29 or younger (LODES variable CA01). This captures the density of young workers near each actor's residence, serving as a proxy for the economic vibrancy and age composition of the local area. The resulting graph has 234 edges.
    \item \textit{Low-earnings employment} ($G_{\text{earn}}$): total number of jobs with monthly earnings of \$1250 or less (LODES variable CE01). This captures the density of low-wage employment near each actor's residence. The resulting graph has 234 edges.
\end{itemize}
Each exposure rate is dichotomized using the same threshold procedure as the crime graph.

We assess the relationship between individuals' weekly activity overlap and their socio-demographic environments using the \texttt{STRUCTURED-RJ} mixture model described in Section \ref{sec:rjmcmc}. We use the reduced-order variant rather than \texttt{STRUCTURED-FP} because with $L=3$ components and $N=11$ weeks, each component is estimated from $N/L \approx 4$ effective observations---the regime where the simulation study demonstrates \texttt{STRUCTURED-RJ} outperforms all alternatives (Table \ref{table:sample_size}). The model assumes that each demeaned sociomatrix $\mathbf{X}_w$ is generated by one of $L=3$ symmetric matrix-variate normal components, each characterized by a distinct demographic similarity graph $G_l$:
\begin{equation}
    \begin{aligned}
        \mathbf{X}_w\ |\ c_w = l &\sim \mathcal{SN}\left(\mathbf{0}, \mathbf{\Omega}_l, a_l\mathbf{I} + b_l\mathbf{\Omega}_l^{m_l} ; G_l\right),
    \end{aligned}
\end{equation}
where the polynomial order $m_l$ is automatically selected via reversible-jump MCMC.

We construct estimates of activity overlap among 40 actors for each of 11 weeks, and we use three measures of socio-demographic exposure similarity: the number of reported crime incidents, the number of youth workers (age 29 or younger), and the number of low-earnings jobs (\$1250/month or less), each within 5km of the actors' estimated home locations. Crime incidence data are sourced from the King County Sheriff's Office historic incident dataset. Youth employment and earnings data are obtained from the Census Bureau's Longitudinal Employer-Household Dynamics Origin-Destination Employment Statistics (LODES), which provide census block-level employment counts by age and earnings for 2018. Each exposure rate $\mathbf{R}_l$ is computed at the individual level and dichotomized into a similarity graph $G_l$ using a threshold on pairwise exposure differences.

We run the \texttt{STRUCTURED-RJ} algorithm for 50,000 iterations after a 10,000 iteration burn-in period. For simplicity, we demean and standardize each sociomatrix $\mathbf{X}_w$ and set $\mathbf{M}_l = \mathbf{0}$.

Figure \ref{fig:posterior_matrix_estimates} displays the posterior mean estimates of $\mathbf{\Omega}_l$ and $\mathbf{\Upsilon}_l$ for each component. The reversible-jump sampler selects $m=1$ for all three components with posterior probability 1.0, indicating that the linear relationship $\mathbf{\Upsilon}_l \approx b_l\mathbf{\Omega}_l$ provides the best fit in this sparse-data setting ($N/L \approx 4$ effective observations per component). This is consistent with the simulation results in Section \ref{sec:simulation}, where \texttt{STRUCTURED-RJ} achieves optimal performance for $m=1$ by reducing the parameter space to 3 effective degrees of freedom per component.

\begin{figure}[ht]
    \centering
    \includegraphics[scale=.3]{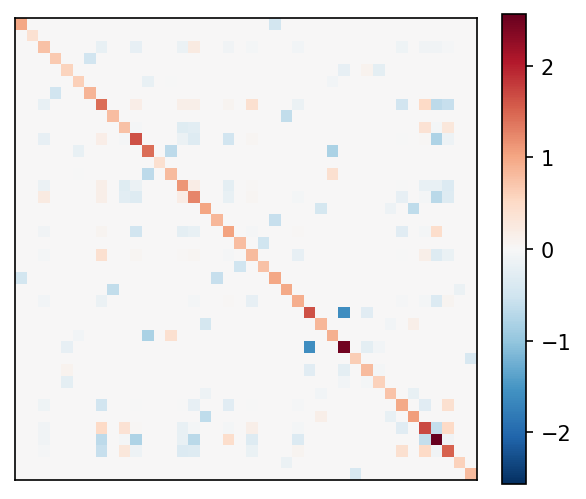}\includegraphics[scale=.3]{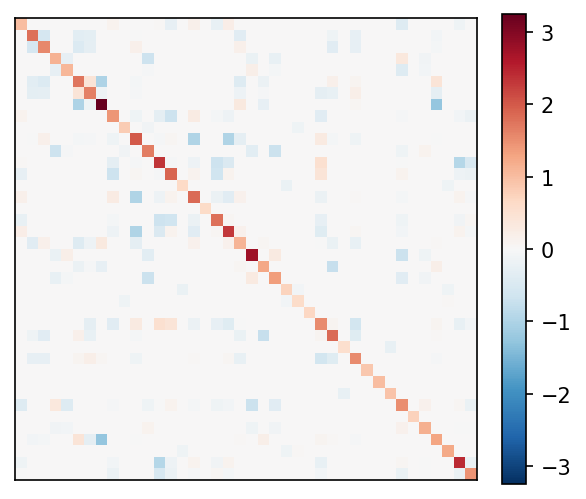}\includegraphics[scale=.3]{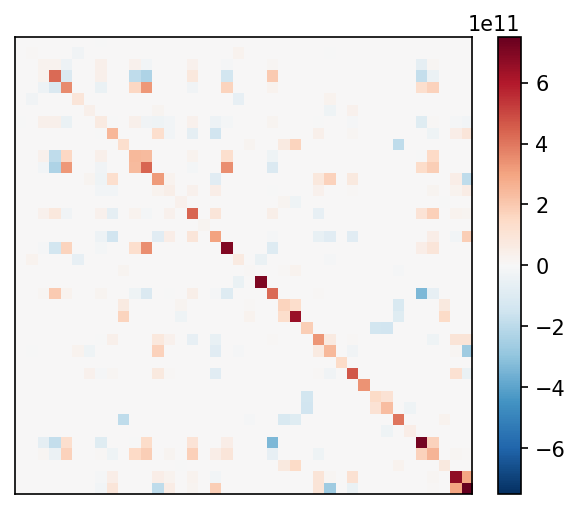}
    \includegraphics[scale=.3]{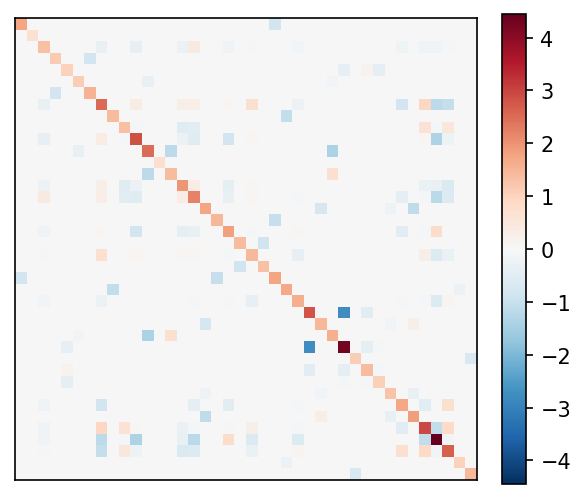}\includegraphics[scale=.3]{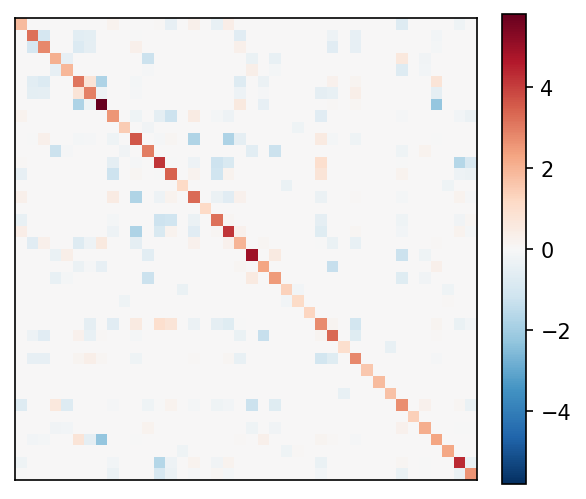}\includegraphics[scale=.3]{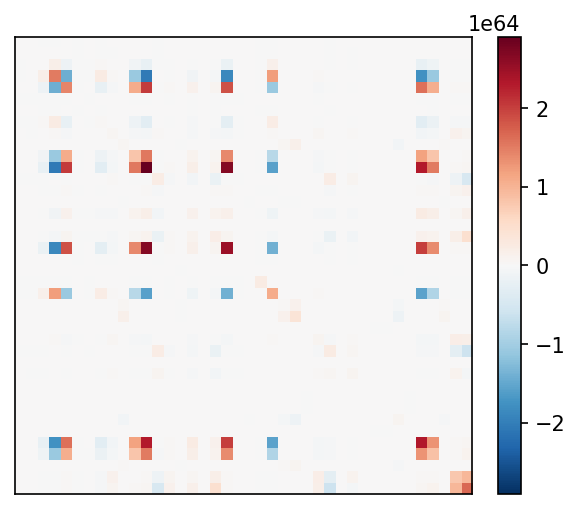}
    \caption{Posterior mean estimate of $\mathbf{\Omega}_l$ (top) and $\mathbf{\Upsilon}_l$ (bottom) for $l =$ crime incidents (left), youth employment (middle), and low-earnings employment (right). Values less than 1e-5 have been truncated to zero for visual clarity.}
    \label{fig:posterior_matrix_estimates}
\end{figure}

Figure \ref{fig:prob_assignments} displays the posterior mean probability of assignment to each similarity graph across 11 weeks using the LFCM-based sociomatrices. Of the 11 weeks, 5 are best described by the crime incidence similarity graph ($\hat{\pi}_{\text{crime}} = 0.43$) and 6 are best described by the youth employment graph ($\hat{\pi}_{\text{youth}} = 0.50$). The near-equal split between crime and youth employment suggests that both local safety environments and economic composition contribute substantially to explaining variation in activity overlap. This finding contrasts with the CPT-based analysis (Table \ref{table:cpt_vs_lfcm}), where the coarser discretization produces 57\% exact zeros that push most weeks toward a single dominant component. The LFCM-based sociomatrices, with their continuous support and higher temporal stability, allow the model to differentiate between weeks where crime environments versus employment demographics better explain mobility similarity.

\begin{figure}[ht]
    \centering
    \includegraphics[scale=.65]{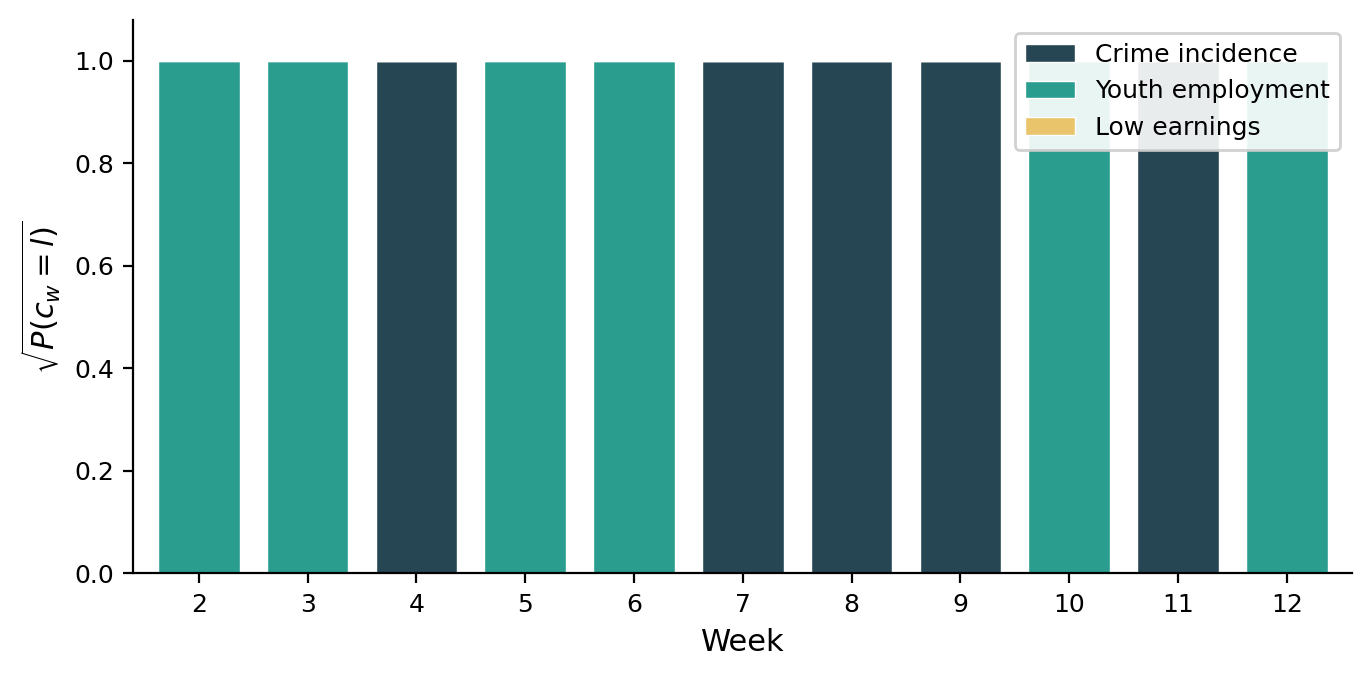}
    \caption{Posterior mean probability of assignment to each graph across 11 weeks.}
    \label{fig:prob_assignments}
\end{figure}

\paragraph{Posterior Estimates} Table \ref{table:application_results} summarizes the posterior estimates for each mixture component using LFCM-based sociomatrices. The reversible-jump sampler selects $m=1$ for all three components, and the near-unity correlations $\text{Corr}(\hat{\mathbf{\Omega}}_l, \hat{\mathbf{\Upsilon}}_l) = 1.00$ for the active components confirm that the linear relationship $\mathbf{\Upsilon}_l \approx b_l\mathbf{\Omega}_l$ is empirically valid. The sparsity of both precision matrices (80--87\% of off-diagonal elements below $10^{-3}$) reflects the sparse nature of the underlying demographic similarity graphs, indicating that the $G$-Wishart prior successfully recovers the conditional independence structure of the data.

\begin{table}[H]
    \centering
    \begin{tabular}{lcccccc}
    & & & & \multicolumn{2}{c}{\small \textbf{Sparsity}} & \\
    \cmidrule{5-6}
    {\small \textbf{Component}} & {\small $\hat{\pi}_l$} & {\small Weeks} & {\small $\hat{m}_l$} & {\small $\mathbf{\Omega}_l$} & {\small $\mathbf{\Upsilon}_l$} & {\small $\text{Corr}(\hat{\mathbf{\Omega}}_l, \hat{\mathbf{\Upsilon}}_l)$}\\
    \midrule
    Crime incidence & .43 & 5 & 1 & .87 & .87 & 1.00 \\
    Youth employment & .50 & 6 & 1 & .80 & .80 & 1.00 \\
    Low earnings & .07 & 0 & 1 & --- & --- & --- \\
    \bottomrule
    \end{tabular}
    \caption{Summary of posterior estimates for each mixture component using LFCM-based sociomatrices. Sparsity denotes the fraction of upper-triangular elements with magnitude below $10^{-3}$. The low-earnings component is effectively inactive.}
    \label{table:application_results}
\end{table}

The nearly equal split between crime ($\hat{\pi} = 0.43$) and youth employment ($\hat{\pi} = 0.50$) suggests that both the local safety environment and the economic composition contribute substantially to explaining the variation in activity overlap. The low-earnings component is effectively inactive ($\hat{\pi} = 0.07$), indicating that this particular demographic dimension does not provide additional explanatory power beyond what is captured by the crime and youth graphs.

Figure \ref{fig:convergence_app} displays the convergence diagnostics for the \texttt{STRUCTURED-RJ} mixture model. The cumulative means of the mixing weights $\hat{\pi}_l$ and component assignment fractions stabilize after approximately 2,000 post-burn-in iterations, indicating adequate mixing. The sampler was run for 50,000 post-burn-in iterations with thinning every 10.
\begin{figure}[ht]
    \centering
    \includegraphics[width=\textwidth]{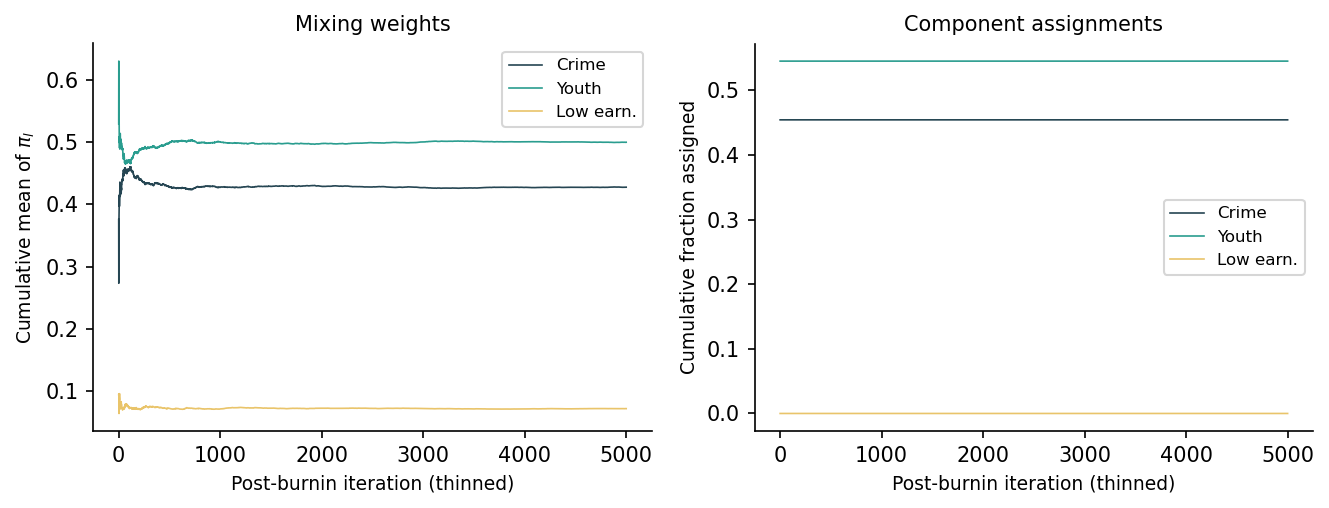}
    \caption{Convergence diagnostics for the application. Left panel: cumulative mean of mixing weights $\pi_l$. Right panel: cumulative fraction of weeks assigned to each component.}
    \label{fig:convergence_app}
\end{figure}

\paragraph{Appropriateness of the LFCM-Based Sociomatrix} Table \ref{table:lfcm_vs_cpt_diagnostics} provides a diagnostic comparison between the sociomatrices based on CPT and LFCM, which justifies the use of LFCM as the primary construction method.

\begin{table}[H]
    \centering
    \begin{tabular}{lcc}
    {\small \textbf{Diagnostic}} & {\small \textbf{CPT}} & {\small \textbf{LFCM}} \\
    \midrule
    Zero/near-zero fraction & .57 & .03 \\
    Consecutive-week rank correlation & .38 & .70 \\
    Active mixture components ($\hat{\pi}_l > .10$) & 2 & 2 \\
    Mixture entropy $H(\hat{\underline{\pi}})$ & .89 & .90 \\
    \bottomrule
    \end{tabular}
    \caption{Diagnostic comparison of sociomatrix construction methods.}
    \label{table:lfcm_vs_cpt_diagnostics}
\end{table}

The LFCM-based sociomatrices are preferable for the symmetric matrix-variate normal model on two grounds. First, CPT produces 57\% exact zeros, violating the continuous support assumption of the Gaussian likelihood; the LFCM's Brownian bridge interpolation yields only 3\% near-zero entries. Second, the consecutive-week rank correlation is 0.70 for LFCM versus 0.38 for CPT, indicating that LFCM produces temporally stable pairwise overlap estimates where the relative ordering of actor pairs is preserved across weeks. This stability is critical: if the sociomatrix is dominated by discretization noise, the mixture model cannot meaningfully distinguish demographic drivers from sampling artifacts.

\paragraph{Comparison with CPT} Table \ref{table:cpt_vs_lfcm} compares the mixture model results under both sociomatrix construction methods. Although the CPT-based analysis assigns the majority of weeks to a single component (crime, $\hat{\pi} = 0.64$), the LFCM-based analysis produces a more balanced split. This difference is attributable to CPT's zero-inflation: the 57\% exact zeros are uniformly uninformative across components, causing the model to concentrate probability mass on whichever graph best explains the sparse nonzero entries. The LFCM sociomatrices, with their continuous support, allow the model to leverage the full pairwise structure and differentiate between weeks where crime versus employment demographics better explain mobility similarity.

\begin{table}[H]
    \centering
    \begin{tabular}{llccc}
     & & \multicolumn{3}{c}{\small \textbf{Weeks assigned}} \\
    \cmidrule{3-5}
    {\small \textbf{Sociomatrix}} & {\small \textbf{Component}} & {\small $\hat{\pi}_l$} & {\small Weeks} & {\small $\hat{m}_l$}\\
    \midrule
    LFCM & Crime & .43 & 5 & 1 \\
    & Youth employment & .50 & 6 & 1 \\
    & Low earnings & .07 & 0 & 1 \\
    \cmidrule{1-5}
    CPT & Crime & .64 & 8 & 1 \\
    & Youth employment & .22 & 2 & 1 \\
    & Low earnings & .14 & 1 & 1 \\
    \bottomrule
    \end{tabular}
    \caption{Mixture model results using LFCM-based versus CPT-based sociomatrices. Both methods select $m=1$ for all components.}
    \label{table:cpt_vs_lfcm}
\end{table}

In particular, both methods unanimously select $m=1$ for all components, strengthening the robustness of the reduced-order parameterization in this sparse-data setting regardless of the sociomatrix construction choice.

\paragraph{Comparison with \texttt{MVGGM}} To quantify the advantage of the commutativity constraint, we compare the \texttt{STRUCTURED-RJ} mixture with independent \texttt{MVGGM} fits \citep{dobra2011bayesian} in LFCM-based sociomatrices. For each graph $G_l$, we fit \texttt{MVGGM} (estimating $\mathbf{\Omega}_l$ and $\mathbf{\Upsilon}_l$ independently) to all 11 weeks, then assign each week to the graph that yields the highest likelihood.

\begin{table}[H]
    \centering
    \begin{tabular}{lccc}
    & \multicolumn{3}{c}{\small \textbf{Weeks assigned}} \\
    \cmidrule{2-4}
    {\small \textbf{Model}} & {\small Crime} & {\small Youth} & {\small Low earnings} \\
    \midrule
    \texttt{STRUCTURED-RJ} & 5 & 6 & 0 \\
    \texttt{MVGGM} & 0 & 2 & 9 \\
    \bottomrule
    \end{tabular}
    \caption{Week assignments under \texttt{STRUCTURED-RJ} versus \texttt{MVGGM} using LFCM-based sociomatrices.}
    \label{table:structured_vs_mvggm_app}
\end{table}

\texttt{MVGGM} assigns 9 of 11 weeks to a single graph (low-earnings), exhibiting the same over-concentration observed with CPT-based sociomatrices. This behavior is explained by the parametric imbalance: \texttt{MVGGM} estimates $p(p+1) = 1640$ free elements per graph (both precision matrices unconstrained), while \texttt{STRUCTURED-RJ} estimates only $\sim$130 per component (graph-constrained Cholesky elements plus 3 scalar parameters). With $N/L \approx 4$ effective observations per component, \texttt{MVGGM} overfits to whichever graph happens to explain the noise in the largest number of weeks, whereas \texttt{STRUCTURED-RJ}'s reduced parameterization enables meaningful differentiation between demographic drivers across time.

\section{Discussion}
\label{sec:discussion}

We introduce \texttt{STRUCTURED}, a Bayesian mixture model designed for inference on symmetric matrix-variate data generated by temporally evolving networks. The model leverages the commutativity property of the symmetric matrix-variate normal distribution to express the column precision matrix $\mathbf{\Upsilon}$ as a polynomial function of the row precision matrix $\mathbf{\Omega}$. We propose two versions: \texttt{STRUCTURED-FP}, which estimates the full polynomial with $p$ coefficients, and \texttt{STRUCTURED-RJ}, which restricts attention to a single nonzero polynomial term, with its order chosen automatically using reversible-jump MCMC. Simulation results indicate that \texttt{STRUCTURED-RJ} is preferable when $N/L$ is small (and outperforms all competing methods when $N \leq 5$ per component), whereas \texttt{STRUCTURED-FP} performs better when $N \geq 10$ per component and the true polynomial relationship is potentially complex. This framework tackles three central challenges in analyzing mobility-based sociomatrices: high dimensionality, temporal evolution, and the requirement to link network structure with external socio-demographic covariates.

Our main methodological contribution is to show that the commutativity constraint of the symmetric matrix-variate normal distribution can be harnessed to achieve practical improvements in estimation. The theoretical results (Theorem \ref{thm:poly2} and Corollary \ref{corr:graphs}) characterize the underlying polynomial structure, and the two \texttt{STRUCTURED} variants offer complementary strategies for exploiting this structure, tailored to different data regimes. Moreover, the spectral representation of the normalizing constant permits efficient evaluation of the likelihood without explicitly constructing the $p^2 \times p(p+1)/2$ transformation matrix $\mathbf{B}_p$.

The analysis of King County mobility data shows that crime incidence and youth employment are the primary demographic dimensions accounting for differences in weekly activity overlap patterns. Across the 11 weeks studied, 5 are most accurately described by the crime incidence similarity graph ($\hat{\pi}_{\text{crime}} = 0.43$), while 6 align best with the youth employment graph ($\hat{\pi}_{\text{youth}} = 0.50$). These findings align with criminological research indicating that neighborhood crime conditions influence mobility through mechanisms of informal social control \citep{crutchfield1982crime}. In contrast, the low-earnings factor is essentially non-contributory, suggesting that this demographic variable does not add further explanatory value.

The unanimous selection of $m=1$ across all components provides empirical support for the hypothesis that, in this setting, the row and column precision matrices share the same conditional independence structure. This finding simplifies interpretation: the graph $G_l$ governing $\mathbf{\Omega}_l$ simultaneously governs $\mathbf{\Upsilon}_l$, meaning that the same demographic similarity relationships determine both the ``who'' and ``how'' of activity overlap.

Several limitations should be noted. First, demographic similarity graphs are constructed by dichotomizing continuous exposure rates, introducing sensitivity to the threshold choice. A fully Bayesian treatment of graph selection \citep{wang2012bayesian} could address this limitation. Second, the model assumes known $L$---extending to nonparametric mixture models via Dirichlet process priors would allow the data to determine the number of components. Third, GPS data represent a convenience sample of mobile device users and may not be representative of the broader population. Fourth, while the LFCM-based sociomatrices substantially improve upon the CPT estimator in terms of Gaussian model compatibility and temporal stability, they inherit the assumptions of the LFCM model itself---in particular, the Brownian motion dynamics within activity regions and the L\'{e}vy flight transitions between them.

Several extensions are natural. The reversible-jump framework could be extended to allow different polynomial orders across mixture components, with a hierarchical prior sharing information across groups. The model could be applied to longitudinal network data beyond mobility, including protein interaction networks, communication networks, or financial correlation matrices where symmetric matrix-variate structure arises naturally. Finally, incorporating temporal dependence across weeks---for example through a hidden Markov model on the component assignments---would capture the observation that mobility patterns exhibit persistence beyond the i.i.d.\ assumption.

\section*{Acknowledgment}

The authors thank Zack Almquist and Abel Rodriguez for helpful discussions. AHW and GSC thank the following university departments for their in-kind support through their affiliate faculty positions: Department of Statistical Sciences and Operations Research at the Virginia Commonwealth University (AHW, GSC), Department of Statistics at the University of Washington (GSC), and the Department of Statistics and Actuarial Science at the University of Waterloo (GSC). The work of AD was partially supported by National Institute of Mental Health grants R01MH131480 and R01MH133488.

\bibliographystyle{plainnat}

\appendix
\section{Proofs}
\label{appendix:proofs}

\subsection{Alternative Parameterization: Spherical Coordinates}
\label{appendix:spherical}

When estimating the full polynomial $\mathbf{\Upsilon} = \sum_{k=0}^{p-1}\alpha_k\mathbf{\Omega}^k$, the coefficient vector $\underline{\alpha}$ admits a scale ambiguity with $\mathbf{\Omega}$. For any constant $c$,
\begin{equation}
    \begin{aligned}
        \sum_k \alpha_k\mathbf{\Omega}\otimes\mathbf{\Omega}^k &= \sum_k \frac{\alpha_k}{c^{k+1}} (c\mathbf{\Omega})\otimes(c\mathbf{\Omega})^k.
    \end{aligned}
\end{equation}
To remedy this, one may constrain $\underline{\alpha} \in \mathcal{S}_{\succ}^{n\text{-}1} \equiv \{\underline{\alpha} \in \mathcal{S}^{n-1}\ \mid\ \sum_{k} \alpha_k \mathbf{\Omega}^k \succ 0\}$, the subspace of the $n\text{-}1$ sphere where $\mathbf{\Upsilon}$ is positive definite.

The coordinates $\underline{\theta} \in \mathbb{R}^{n\text{-}1}$ for the $n\text{-}1$ sphere can be mapped to euclidean coefficients $\underline{\alpha} \in \mathbb{R}^{n}$ using equations
\begin{equation}
    \begin{aligned}
        \alpha_k &= \cos(\theta_k)^{\mathbbm{1}\{k < n\text{-}1\}}\sin(\theta_k)^{\mathbbm{1}\{k=n-1\}}\left[\prod_{i=0}^{k-1} \sin(\theta_i)\right]^{\mathbbm{1}\{k>0\}},
    \end{aligned}
\end{equation}
where $\theta_0,\ldots,\theta_{n-2} \in [0,\pi]$ and $\theta_{n-1}\in [0,2\pi)$. The inverse transformation is given by
\begin{equation}
    \begin{aligned}
        \theta_k &= \left[\cot^{-1}\dfrac{\alpha_k}{\|\underline{\alpha}_{(k+1):(n-1)}\|_2}\right]^{\mathbbm{1}\{k < n - 2\}} \left[2\cot^{-1}\dfrac{\alpha_k + \|\underline{\alpha}_{(k+1):(n-1)}\|_2}{\alpha_{n-1}}\right]^{\{k=n-2\}}.
    \end{aligned}
\end{equation}
The determinant of the Jacobian is
\begin{equation}
    \begin{aligned}
        |J_n| &= \prod_{k=0}^{n-2}\sin^{n-1-k}(\theta_k).\\
    \end{aligned}
\end{equation}

The conditional restriction of $\theta_j$ given $\theta_{-j}$ to ensure positive definiteness is derived as follows. Let $\lambda$ denote an eigenvalue of $\mathbf{\Omega}$ and $\underline{b} \equiv (1,\lambda, \lambda^2,\ldots,\lambda^{d\text{-}1})$. The positive definiteness of $\mathbf{\Upsilon}$ requires $\sum_{k} b_k\alpha_k > 0$, which in spherical coordinates yields
\begin{equation}
    \begin{aligned}
        \theta_j < \pi \wedge \cos^{-1}\left(-\frac{\sum_k b_k\alpha_k}{b_jc_j}\right), && b_jc_j > 0\\
        \theta_j > 0 \vee \cos^{-1}\left(-\frac{\sum_k b_k\alpha_k}{b_jc_j}\right), && b_jc_j < 0\\
    \end{aligned}
\end{equation}
for $c_j \equiv \left[\prod_{i=0}^{j-1} \sin(\theta_i)\right]^{\mathbbm{1}\{j>0\}}$. These bounds define the constrained proposal region for a random walk on $\mathcal{S}_{\succ}^{n\text{-}1}$. The Metropolis-Hastings proposal density under the change of variables is $f_{\underline{\alpha}}(\underline{\alpha}) = f_{\underline{\theta}}(\underline{\theta})\left|\prod_{k=1}^{n-2}\sin^{n-k-1}(\theta_k)\right|^{-1}$.

While this parameterization resolves the scale ambiguity in the full polynomial model, the reversible-jump approach described in Section \ref{sec:rjmcmc} achieves identifiability more simply through the normalization $\mathbf{\Omega}_{1,1} = 1$ combined with the reduced-order parameterization $\mathbf{\Upsilon} = a\mathbf{I} + b\mathbf{\Omega}^m$.

\subsection{Proofs of Theorems}
\label{appendix:proofs_theorems}

\begin{lemma}
Suppose $\mathbf{\Omega\Upsilon} = \mathbf{\Upsilon\Omega}$. Then $\mathbf{\Omega}$ and $\mathbf{\Upsilon}$ are simultaneously diagonalizable.
\begin{proof}
Decompose $\mathbb{R}^{p\times p}$ into the direct sum of eigenspaces $\mathbf{E}_{\lambda_1} \oplus \mathbf{E}_{\lambda_2}\oplus \cdots \oplus \mathbf{E}_{\lambda_d}$ for $d \leq p$, $\lambda_i$ the eigenvalues of $\mathbf{\Omega}$, and $\mathbf{E}_{\lambda_i}$ the eigenspaces of $\mathbf{\Omega}$. If $\mathbf{\Omega}\underline{u}_i = \lambda_i\underline{u}_i$, then
\begin{equation}
    \nonumber
    \begin{aligned}
        \mathbf{\Omega}(\mathbf{\Upsilon}\underline{u}_i) = (\mathbf{\Omega\Upsilon})\underline{u}_i = (\mathbf{\Upsilon\Omega})\underline{u}_i = \mathbf{\Upsilon}(\mathbf{\Omega}\underline{u}_i) = \mathbf{\Upsilon}(\lambda_i\underline{u}_i) = \lambda_i\mathbf{\Upsilon}\underline{u}_i.
    \end{aligned}
\end{equation}
Hence $\mathbf{\Upsilon}\underline{u}_i$ is an eigenvector of $\mathbf{\Omega}$, so $\mathbf{\Upsilon}$ preserves the eigenspace $\mathbf{E}_{\lambda_i}$. Now decompose each $\mathbf{E}_{\lambda_i}$ into the direct sum of eigenspaces for $\mathbf{\Upsilon}$. The sum of these eigenspaces across all $i$ provides a mutual basis for $\mathbf{\Omega}$ and $\mathbf{\Upsilon}$.
\end{proof}
\end{lemma}

\begin{thm}
Suppose $\mathbf{\Sigma}, \mathbf{\Psi}$ are square invertible matrices in $\mathbb{R}^{p\times p}$ such that $\mathbf{\Sigma}\mathbf{\Psi} = \mathbf{\Psi}\mathbf{\Sigma}$. Let $\mathbf{\Omega} = \mathbf{\Sigma}^{-1}$ and $\mathbf{\Upsilon} = \mathbf{\Psi}^{-1}$. Then
$$\mathbf{\Omega}\mathbf{\Upsilon} = \mathbf{\Upsilon}\mathbf{\Omega}.$$
\end{thm}

\begin{proof}
\begin{equation}
    \begin{aligned}
        \mathbf{\Omega}\mathbf{\Upsilon} &\equiv \mathbf{\Sigma}^{-1}\mathbf{\Psi}^{-1}\\
        &= (\mathbf{\Psi}\mathbf{\Sigma})^{-1}\\
        &= (\mathbf{\Sigma}\mathbf{\Psi})^{-1}\\
        &= \mathbf{\Psi}^{-1}\mathbf{\Sigma}^{-1}\\
        &\equiv \mathbf{\Upsilon}\mathbf{\Omega},
    \end{aligned}
\end{equation}
where ``$\equiv$'' denotes equivalence by definition.
\end{proof}

\begin{lemma}
Let $\mathbf{\Omega}$, $\mathbf{\Upsilon} \in \mathbb{R}^{p\times p}$ be real symmetric positive definite matrices. Suppose $\mathbf{\Omega}\mathbf{\Upsilon} = \mathbf{\Upsilon}\mathbf{\Omega}$. Then $\mathbf{\Omega}\otimes\mathbf{\Upsilon} = \mathbf{\Upsilon}\otimes\mathbf{\Omega}$. 

\begin{proof}
By Lemma  \ref{lemma:simul_diag}, since $\mathbf{\Omega},\mathbf{\Upsilon}$ commute they are simultaneously diagonalizable. Let $\mathbf{U}$ be an orthogonal matrix such that $\mathbf{\Omega} = \mathbf{U}\mathbf{\Lambda}\mathbf{U}^\top$ and $\mathbf{\Upsilon} = \mathbf{U}\mathbf{\Delta}\mathbf{U}^\top$. Then

\begin{equation}
\nonumber
\begin{aligned}
\mathbf{\Omega}\mathbf{\Upsilon} &= \mathbf{U\Lambda}\mathbf{U^\top U\Delta}\mathbf{U^\top} = \mathbf{U\Lambda}\mathbf{\Delta}\mathbf{U^\top}\\
=& \mathbf{U\Delta}\mathbf{\Lambda}\mathbf{U^\top} = \mathbf{U\Delta}\mathbf{U^\top U\Lambda}\mathbf{U^\top} = \mathbf{\Upsilon}\mathbf{\Omega}.
\end{aligned}
\end{equation}

By the mixed product property of matrix products,
\begin{equation}
\nonumber
\begin{aligned}
\mathbf{\Omega}\otimes\mathbf{\Upsilon} &= 
(\mathbf{U}\mathbf{\Lambda}\mathbf{U}^\top)\otimes(\mathbf{U}\mathbf{\Delta}\mathbf{U}^\top)\\
&= (\mathbf{U}\mathbf{\Lambda}\otimes\mathbf{U}\mathbf{\Delta})(\mathbf{U}^\top\otimes\mathbf{U}^\top)\\
&=(\mathbf{U}\otimes\mathbf{U})(\mathbf{\Lambda}\otimes\mathbf{\Delta})(\mathbf{U}^\top\otimes\mathbf{U}^\top)\\
&= (\mathbf{U}\otimes\mathbf{U})(\mathbf{\Delta}\otimes\mathbf{\Lambda})(\mathbf{U}^\top\otimes\mathbf{U}^\top)\\
&=  (\mathbf{U}\mathbf{\Delta}\otimes\mathbf{U}\mathbf{\Lambda})(\mathbf{U}^\top\otimes\mathbf{U}^\top)\\
&= (\mathbf{U}\mathbf{\Delta}\mathbf{U}^\top)\otimes(\mathbf{U}\mathbf{\Lambda}\mathbf{U}^\top)\\
&= \mathbf{\Upsilon}\otimes\mathbf{\Omega}.
\end{aligned}
\end{equation}
\end{proof}
\end{lemma}

\begin{lemma}
Suppose $\mathbf{\Omega}, \mathbf{\Upsilon}$ are simultaneously diagonalizable, and $\mathbf{\Omega}$ has only simple roots. Then $\mathbf{\Upsilon} = \sum_{k=0}^{p-1} \alpha_k\mathbf{\Omega}^k$.

\begin{proof}
Since $\mathbf{\Omega}$ and $\mathbf{\Upsilon}$ are simultaneously diagonalizable and $\lambda_1,\ldots, \lambda_p$ are distinct, there exists a polynomial $p(t) = \sum_{k=0}^{p-1} \alpha_k t^k$ such that $p(\lambda_i) = \delta_i$ for all $i=1,\ldots, p$. Hence
\begin{equation}
    \nonumber
    \begin{aligned}
        p(\mathbf{\Omega})\underline{u}_i &= \sum_{k=0}^{p-1} \alpha_k\mathbf{\Omega}^k\underline{u}_i = \sum_{k=0}^{p-1} \alpha_k \lambda_i^{k}\underline{u}_i = p(\lambda_i)\underline{u}_i = \delta_i\underline{u}_i,
    \end{aligned}
\end{equation}
i.e.\ $p(\mathbf{\Omega})$ and $\mathbf{\Upsilon}$ share the same eigenvalues. Moreover, since $\mathbf{\Omega}$ and $\mathbf{\Upsilon}$ share the same eigenvectors, $p(\mathbf{\Omega})$ and $\mathbf{\Upsilon}$ share the same eigenvectors, so $p(\mathbf{\Omega}) = \mathbf{\Upsilon}$.
\end{proof}
\end{lemma}

\begin{thm}
Let $\mathbf{\Omega},\mathbf{\Upsilon}$ be defined as above, and suppose $\mathbf{\Omega} \sim \text{Wish}_G(\delta, \mathbf{I})$, and $\mathbf{\Upsilon}$ commutes with $\mathbf{\Omega}$. Then 
$$\mathbf{\Upsilon} = \lim_{\varepsilon \to 0} \sum_{k=0}^{p-1} \alpha_k(\mathbf{\Omega} + \varepsilon)^k$$
for some $\alpha_0, \ldots, \alpha_{p-1} \in \mathbb{R}$. If $\mathbf{\Omega}$ has only simple roots, then one can take $\varepsilon = 0$.
\end{thm}

\begin{proof}
Note that if $\mathbf{\Omega}$ has only simple roots, the result follows from Lemma \ref{lemma:poly}. Otherwise, since $\mathbf{\Omega}$ is Hermitian it may be approximated to arbitrary accuracy by a Hermitian matrix with a simple spectrum \citep{denton2022eigenvectors}. 
\end{proof}

\begin{thm}
Let $\mathbf{\Sigma},\mathbf{\Psi},\mathbf{\Omega},\mathbf{\Upsilon}$ be defined as above. Then for $\alpha_0,\ldots,\alpha_{p-1} \in \mathbb{R}$:
$$P\left(\mathbf{\Upsilon} = \sum_{k=0}^{d-1} \alpha_k\mathbf{\Omega}^k\right) = 1.$$
\end{thm}

\begin{proof}
Since $\mathbf{\Omega}$ is a real symmetric matrix, it has a spectral decomposition and is thus diagonalizable. By definition, a matrix has distinct eigenvalues if and only if the discriminant of its characteristic polynomial is nonzero. The characteristic polynomial is a polynomial function of the entries of the matrix; hence its roots have zero measure for any reasonable distribution.
\end{proof}

\begin{thm}
Let $\mathbf{\Sigma},\mathbf{\Psi},\mathbf{\Omega},\mathbf{\Upsilon}$ be defined as above. Suppose $\mathbf{X}$ follows a symmetric matrix-variate normal distribution. Then the partial correlations associated with $\mathbf{X}$ are $0$ if and only if
$$\sigma_{ik}\psi_{jl} + \sigma_{jk}\psi_{il} + \sigma_{il}\psi_{jk} + \sigma_{jl}\psi_{ik} = 0.$$
\end{thm}

\end{document}